\documentclass{article}
\usepackage{amsthm}
\usepackage{amsmath}
\usepackage{amsfonts}
\usepackage{mathtools}
\usepackage[all]{xy}
\usepackage{tikz}
\usetikzlibrary{arrows}
\usetikzlibrary{decorations.markings}
\usetikzlibrary{shapes}
\tikzset{->-/.style={decoration={
  markings,
  mark=at position .5 with {\arrow{>}}},postaction={decorate}}}

\tikzset{-->-/.style={decoration={
  markings,
  mark=at position .7 with {\arrow{>}}},postaction={decorate}}}

\usepackage{mathabx}
\usetikzlibrary{shapes}

\newcommand{\rats}{\mbox{\(\mathbb Q\)}}

\newtheorem{definition}{Definition}

\newtheorem{theorem}[definition]{Theorem}
\newtheorem{lemma}[definition]{Lemma}

\newtheorem{proposition}[definition]{Proposition}
\newtheorem{problem}[definition]{Problem}

\newcommand{\dom}[1]{\mathrm{dom}(#1)}
\newcommand{\s}{\mathcal{S}}

\newcommand{\dleq}{\sqsubseteq}

\def\Los{\L{o}\'{s}}

\def\restr #1{{\restriction_{#1}}}

\usepackage[utf8]{inputenc}
\def\set#1{{ \{#1\}}}
\def\c#1{{\mathcal #1}}

\def\btl{\blacktriangleleft}

\def\nb#1{{$\bullet$\marginpar{$\bullet$ #1}}}
\title{Finite Representability of Semigroups with Demonic Refinement}
\author{Robin Hirsch and Ja\v s \v Semrl}
\date{September 2020}
\begin{document}

\maketitle

\begin{abstract}
Composition and demonic refinement $\sqsubseteq$ of binary relations are defined by
\begin{align*}
(x, y)\in (R;S)&\iff \exists z((x, z)\in R\wedge (z, y)\in S)\\
    R\sqsubseteq S&\iff (dom(S)\subseteq dom(R) \wedge R\restr{dom(S)}\subseteq S)
    \end{align*}
where $dom(S)=\set{x:\exists y (x, y)\in S}$ and  $R\restr{dom(S)}$ denotes the restriction of $R$ to pairs $(x, y)$ where $x\in dom(S)$.
    Demonic calculus  was introduced to model the total correctness of non-deterministic programs and has been applied to program verification \cite{berghammer1986relational,dijkstrapred}.
    
We prove that the class $R(\sqsubseteq, ;)$ of abstract $(\leq, \circ)$ structures isomorphic to a set of binary relations ordered by demonic refinement with composition cannot be axiomatised by any finite set of first-order $(\leq, \circ)$ formulas.  We provide a fairly simple, infinite, recursive axiomatisation that defines $R(\sqsubseteq, ;)$.   We prove that a finite representable $(\leq, \circ)$ structure has a representation over a finite base. This appears to be the first example of a signature for binary relations with composition where the representation class is non-finitely axiomatisable, but where the finite representations for finite representable structures property holds.
  
\end{abstract}

\section{Introduction and Motivation}
The simplest way of representing a $(\leq, \circ)$ structure is to interpret the binary relation $\leq$ as set inclusion $\subseteq$, and the binary function $\circ$ as composition of binary relations $;$.   The class $R(\subseteq, ;)$ of  abstract $(\leq, \circ)$ structures isomorphic to sets of binary relations with inclusion and composition is defined exactly by the axioms of ordered semigroups \cite{zareckii1959representation}, i.e. associativity, partial order, left and right monotonicity.  It is clear that these axioms are valid over $R(\subseteq, ;)$.  Conversely,  given an ordered semigroup $\s=(S, \leq, \circ)$ we may extend the structure to the ordered semigroup  $\s'= (S', \leq, \circ)$ by adding a single new two-sided identity element $e$ where $e\not\leq s$ and $s\not\leq e$ for $s\in S$, and then defining a representation $\theta$ of $\s$  over $\s'$ by
 \[ (x, y)\in s^\theta\iff y\leq x\circ s\]
illustrated in the first diagram of Figure~\ref{fig:1}. The extra identity element is used to prove faithfulness of $\theta$: if $s\not\leq t\in\s$ then $(e, s)\in s^\theta\setminus t^\theta$.    A dual representation $\theta'$ of $\c S$ over $\c S'$,  illustrated in the second part of Figure~\ref{fig:1}, is defined by
 \[ (x, y)\in s^{\theta'}\iff x\leq s\circ y.\]
\begin{figure}
\[
\xymatrix{
\bullet\ar^s[rr]&&\bullet  & &\bullet &\\
&\\
&\bullet\ar^x[luu]\ar_y[ruu]&&\bullet\ar^s[rr]\ar^x[ruu]&&\bullet\ar_y[luu]}
\]
 \caption{\label{fig:1}Two representations, for ordered semigroups}
 \end{figure}
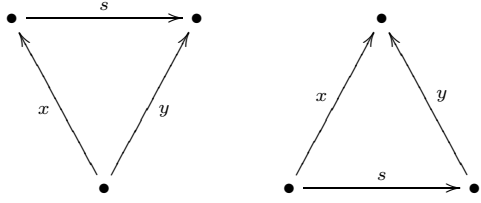
 Both inclusion and composition have demonic variants, written $(\sqsubseteq, *)$ called demonic refinement (defined above in the abstract) and demonic composition defined by
 \[ R*S=R;S\;\cap\;\set{(x, y):\forall z((x, z)\in R\rightarrow z\in \dom{S})}.\] 
 Closely related to the demonic refinement relation is the demonic join operator $\sqcup$, defined by
\[ R\sqcup S = (R\cup S)\restr{\dom{R}\cap \dom{S}}\]
In the set of all binary relations over a set $X$, \/ $R\cup S$ is the meet (least upper bound) of $R, S$ with respect to $\sqsubseteq$.  Conversely, given the operator $\sqcup$ we may recover the relation $\sqsubseteq$ by defining $R\sqsubseteq S\iff R\sqcup S=S$.  Note, however, that an operator $\sqcap$ that returns the greatest lower bound of two relations may not be defined, it is not in general the case that two binary relations have any common lower bound with respect to $\sqsubseteq$.

 It is known that $R(\sqsubseteq, *)$ is also axiomatised by the three axioms of ordered semigroups.  For $R(\subseteq, *)$, although $*$ remains 
 associative,  $\subseteq$ is still a partial order and $*$ is right monotonic with respect to $\subseteq$, we find that left monotonicity fails.  Recently  it was shown that $R(\subseteq, *)$ is not finitely axiomatisable \cite{HMS:20}.
 
 The case we focus on here is $R(\sqsubseteq, ;)$, with ordinary composition and demonic refinement.  This time, we find that both left and right monotonicity fail.  How do we axiomatise $R(\sqsubseteq, ;)$?  Of course we may use the axioms of associativity and partial order, but what additional axioms should be included in order to fill the gap created by the omission of the two monotonicity axioms, and define the representation class?  The main results here are  a recursively defined  infinite set of axioms that define $R(\sqsubseteq, ;)$ and a proof that no finite set of axioms can do it.    Our proof that the recursive axiomatisation is complete also shows that a finite representable $(\sqsubseteq, ;)$-structure has a representation over a finite base set.

Algebras of binary relations have been used extensively to model program semantics \cite{mili1987relational, dijkstrapred}, and the introduction of demonic choice ($\sqcup$) and demonic composition ($*$) has extended this framework towards reasoning about the total correctness of non-deterministic Turing Machines \cite{dijkstrapred, berghammer1986relational}.

The introduction of the demonic refinement predicate led to further verification applications, for example utilising Refinement Algebras \cite{von2004towards, de2008structure}. Furthermore, relaxing the requirement that composition is a total binary operator we obtain  \emph{refined semigroupoids}, which have been of interest in relation-algebraic programming \cite{kahl2008relational}.

The fairly extensive literature on demonic relations and operators includes a variety of different notations.  In the context of Kleene Algebra extensions, such as Refinement Algebra, where the emphasis is on the behaviour of tests, \/ $\sqcap, \sqsupseteq$ are sometimes used in place of $\sqcup, \sqsubseteq$.

\section{Axiomatising $R(\sqsubseteq, ;)$}
We focus on the signature $(\sqsubseteq, ;)$, in the abstract case the corresponding symbols will be $(\leq, \circ)$. 
A binary relation over the base $X$ is a subset of $X\times X$.   A \emph{concrete} $(\sqsubseteq, ;)$ structure is a set of binary relations over some base, closed under  composition, with demonic refinement.  An isomorphism from an abstract $(\leq, \circ)$ structure to a concrete $(\sqsubseteq, ;)$ structure is called a representation.  $R(\sqsubseteq, ;)$ denotes the class of all $(\leq, \circ)$ structures isomorphic to concrete $(\sqsubseteq, ;)$ structures.

Given a $(\leq, \circ)$ structure $\s$ we let $\s'$ be the structure obtained from $\s$ by adjoining a single  new identity element $e$ where $e\circ x=x=x\circ e,\; e\leq e=e\circ e$ but $x\not\leq e$ and $e\not\leq x$ for $x\in\s$.

The signature does not include the domain operation, nor does it include `angelic' (ordinary) set inclusion. However, we will define with infinitary $(\leq, \circ)$-formulas, the predicates $\blacktriangleleft, \triangleleft^s$ to signify the domain inclusion and inclusion of the restriction to the domain of $s$ respectively.  

Let

$$a \blacktriangleleft b \Leftrightarrow \bigvee_{n < \omega} a \blacktriangleleft_n b$$
$$a \triangleleft^s b \Leftrightarrow \bigvee_{n < \omega} a \triangleleft^s_n b$$
where
\begin{align*}
    a \blacktriangleleft_0 b &\Leftrightarrow 
    a\geq b\vee \exists c (a \geq b\circ c)\\
    a\triangleleft_0^s b&\Leftrightarrow (a\leq b\wedge s= b)\\
   a \blacktriangleleft_{n+1} b &\Leftrightarrow \left\{ \begin{array}{l}
    (a \triangleleft^a_n b)\;\vee\\
     \exists c\;( a \blacktriangleleft_n c\wedge c \blacktriangleleft_n b)\;\vee\\
\exists d, f,f'\;(a=d\circ f\wedge f\blacktriangleleft_nf'\wedge b=d\circ f'))
   \end{array}\right\}\\
  a \triangleleft^s_{n+1} b &\Leftrightarrow  \left\{\begin{array}{l}
    (\exists c\; (a \triangleleft^s_n c\wedge c \triangleleft^s_n b)) \;\vee \\
 \exists c,c',d,d'\;(a=c\circ d\wedge c\triangleleft_n^s c'\wedge d\triangleleft^d_n d'\wedge b=c'\circ d')\;\vee\\
 \exists s'(a\triangleleft^{s'}_n b\wedge s\btl_n s')
\end{array}\right\}
\end{align*}

\begin{lemma}\label{lem:trans}    $\;$
\begin{enumerate}
\item   \label{en:trans}  
 $a\btl_n b\wedge b\btl_nc\rightarrow a\btl_{n+1}c,\;a\triangleleft_n^s b\wedge b\triangleleft_n^s c\rightarrow a\triangleleft^s_{n+1}c$, so  $\blacktriangleleft$ and $\triangleleft^s$ are transitive, for each $s\in\sc S$

\item\label{en:tx} $a\triangleleft_n^sb\circ c,\; b\triangleleft_n^s b',\; c\triangleleft_n^cc'$ implies $a\triangleleft_{n+1}^sb'\circ c'$,
\item\label{en:bx}$d\btl_n a\circ c,\; a\leq a',\; d\btl_n a',\; c\btl_n c'$ implies $d\btl_{n+3} a'\circ c'$,
\item \label{en:ss'} $s\blacktriangleleft_n s',\; a\triangleleft_n^{s'}b$ implies  $a\triangleleft_{n+1}^s b$
\end{enumerate}
\end{lemma}

\begin{proof}
\eqref{en:trans}, \eqref{en:tx}, \eqref{en:ss'} follow directly from the definitions of $\btl, \triangleleft$. For \eqref{en:bx}, observe how from $a \leq a'$ we have $a \triangleleft^{a'}_0 a'$ which, together with $a \circ c' \btl_0 a'$, give us $a \triangleleft^{a \circ c'}_1 a'$. From this and $c' \triangleleft^{c'}_0 c'$ we get $a \circ c' \triangleleft^{a \circ c'}_2 a' \circ c'$ and thus $a \circ c' \btl_3 a' \circ c'$. We also have $c \btl_n c'$ and hence $a \circ c \btl_{n+1} a \circ c'$. So, by the transitive steps $d \btl_n a \circ c \btl_{n+1} a \circ c' \btl_3 a' \circ c'$ we obtain $d\btl_{n+3} a'\circ c'$.
\end{proof}

\begin{lemma}
\label{lem:PredicatesDefinedWell}
Let  $\mathcal{S} \in R(\dleq, ;)$ and let $\theta$ be a representation of $\c S$. For all $a,b,s \in \mathcal{S}$ 
\begin{align*}
    a \blacktriangleleft b &\Rightarrow \dom{a^\theta} \subseteq \dom{b^\theta}\mbox{, and}\\
    a \triangleleft^s b &\Rightarrow a^\theta\restr{\dom{s^\theta}} \subseteq b^\theta.
\end{align*}
\end{lemma}

\begin{proof}
We prove the claim by induction over $n$. In the base case,  if $a\blacktriangleleft_0 b$ then either $a^\theta\sqsupseteq b^\theta$ or $a^\theta \sqsupseteq b^\theta;c^\theta$ (for some $c$)  hence $\dom{a^\theta} \subseteq \dom{b^\theta}$.  And if $a\triangleleft^s_0 b$ then  $s=b,\; a\leq b$, so $a^\theta\restr{\dom{s^\theta}}= a^\theta\restr{\dom{b^\theta}}\subseteq b^\theta$.

For the inductive step, suppose $a\blacktriangleleft_{n+1} b$, from the recursive definition, there are three alternatives. In  the first case, $a\triangleleft_n^a b$ then inductively $a^\theta=a^\theta\restr{\dom{a^\theta}}\subseteq b^\theta$ so $\dom{a^\theta}\subseteq \dom{b^\theta}$.  In the second case,  inductively $\dom{a^\theta}\subseteq \dom{c^\theta}\subseteq \dom{b^\theta}$. In the third case,  there are $d,  f,f'$ where $a=d\circ f,\; f\blacktriangleleft_n f'$ and $b=d\circ f'$.  For any  $x\in \dom{a^\theta}$, there is $y$ such that $(x, y)\in a^\theta$ and there is $z$ such that $(x, z)\in d^\theta,\; (z, y)\in f^\theta$.    Inductively, $z\in \dom{f^\theta}\subseteq \dom{(f')^\theta}$ so there is $w$ such that $(z, w)\in (f')^\theta$, hence $(x, w)\in d^\theta;(f')^\theta=b^\theta$, so $x\in \dom{b^\theta}$, proving $\dom{a^\theta}\subseteq \dom{b^\theta}$.    

 Now suppose $a\triangleleft^s_{n+1}b$.  There are three  alternatives in the recursive definition. 
In the first case, inductively  $a^\theta\restr{\dom{s^\theta}}\subseteq c^\theta$ and $c^\theta\restr{\dom{s^\theta}}   \subseteq b^\theta$, so $a^\theta\restr{\dom{s^\theta}}\subseteq b^\theta$.  In the second case, there are $c, c', d, d'$ as in the definition.  If $x\in \dom{s^\theta}$ and $(x, y)\in a^\theta$ then there is $z$ such that $(x, z)\in c^\theta,\;(z, y)\in d^\theta$.  Inductively, $(x, z)\in (c')^\theta$ and $(z, y)\in (d')^\theta$, hence $(x, y)\in (c'\circ d')^\theta=b^\theta$.   In the third case, $\dom{s^\theta}\subseteq \dom{(s')^\theta}$, so  $a\restr{\dom{s^\theta}}\subseteq a\restr{\dom{(s')^\theta}}\subseteq b^\theta$. This proves $a^\theta\restr{\dom{s^\theta}}\subseteq b^\theta$, as required.
\end{proof}

Let
\begin{align*}
\sigma_n &= ((b\blacktriangleleft_n a\wedge a\triangleleft^b_n b)\rightarrow a\leq b)\\
\sigma & =((b\blacktriangleleft a \wedge a \triangleleft^b b)\rightarrow a\leq b)
\end{align*}
For finite $n$, $\sigma_n$ is a first-order formula, while $\sigma$ is infinitary and is equivalent to $\bigwedge_{n<\omega} \sigma_n$.

\begin{lemma}\label{lem:sound}
\[ R(\sqsubseteq, ;)\models \sigma.\]
\end{lemma}

\begin{proof}
Let $\c S\in R(\sqsubseteq, ;)$ and let $\theta$ be a representation.    
Assume the premise of $\sigma$, \/
$\c S\models (b\blacktriangleleft a\wedge a\triangleleft^b b)$.
By the previous Lemma, 
$\dom{b^\theta}\subseteq \dom{a^\theta}$ and 
$a^\theta\restr{\dom{b^\theta}}\subseteq b^\theta$,   i.e.  $a^\theta\sqsubseteq b^\theta$.   Since $\theta$ represents $\leq$ as $\sqsubseteq$ we must have $\s\models a\leq b$.
Thus $\c S\models\sigma$.
\end{proof}

The following definition is used to prove completeness of our axioms.
\begin{definition}\label{def:theta}  Let $\c S$ be a $(\leq, \circ)$-structure.  Consider the base set 
\[ X= X_i\stackrel{\bullet}{\cup} X_f \stackrel{\bullet}{\cup} X_\beta\]
where
\begin{align*}
    X_i&=\set{(d_i, s_i):d, s\in\s,\; d\blacktriangleleft s}\\
    X_f&=\set{s_f:s\in\s}\\
    X_\beta&=\set{d_\beta:d\in\s},
\end{align*}
we refer to the points in $X_i, X_f, X_\beta$ as \emph{initial points}, \emph{following points} and \emph{branch points}, respectively, it may help to visualise these points using Figure~\ref{fig:points}.

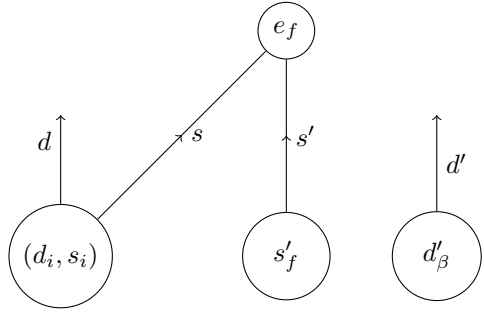
\begin{figure}
    \centering
    \begin{tikzpicture}

            \node[draw, circle] (oi) at (0,0) {$(d_i, s_i)$};    
            \node (ti) at (0,2) {};
            \node[draw,circle] (e) at (3,3) {$e_f$};
     \path (oi) edge[->] node[above left]{$d$} (ti) ;
     
          \path (oi) edge[->-] node[ right]{$s$} (e);
          
          \node[draw,circle] (of) at (3,0) {$\;\;s'_f\;\;$};
          \node[draw,circle] (ob) at (5,0){$\;\;d'_\beta\;\;$};
          \path (of) edge[->-] node[ right]{$s'$} (e);
          \node (tb) at (5, 2){};
          \path (ob) edge[->] node[right]{$d'$}(tb);
     
    \end{tikzpicture}
    \caption{Points  $(d_i, s_i)\in X_i,\;\ s'_f\in  X_f,\; d'_\beta\in X_\beta$}
    \label{fig:points}
\end{figure}

For each $x\in X$ let $\delta(x)\in\s$ be defined by $\delta(x)=d$ if and only if $x\in\set{{(d_i, s_i), d_f, d_\beta:d\blacktriangleleft s\in\s}}$. It may help the reader to think of $\delta(x)$ as an element with $\btl$-minimal domain in the representation we construct, illustrated as a vertical outgoing arrow in Figure~\ref{fig:points}.
For $x\in  X_i\stackrel{\bullet}{\cup} X_f$, define $\lambda(x)\in\s$ by letting $\lambda(x)=s$ iff $x\in\set{(d_i, s_i), s_f: d\blacktriangleleft s}$ 
(undefined if $x\in X_\beta$), illustrated as the label of the edge $(x, e_f)$ in Figure~\ref{fig:points}.   $\lambda(x)$ will be used as the $\triangleleft^{\delta(x)}$-minimal label of the edge $(x, e_f)$ when $x\in X_i\cup X_f$, there are no labels on $(x, e_f)$ when $x\in X_\beta$.   Observe $\delta(x)=\lambda(x)$ for $x\in X_f$.

For each $a\in\s$ define a binary relation $a^\theta\subseteq X\times X$ by letting $(x, y)\in a^\theta$ if and only if 
\begin{enumerate}\renewcommand{\theenumi}{\Roman{enumi}}
    \item $y\not\in X_i$,\label{en:i}
    \item $x\in X_\beta \Rightarrow y\in X_\beta$, \label{en:b}
     \item $\delta(x)\blacktriangleleft a\circ\delta(y)$ and\label{en:delta}
      \item $x\in X_i\cup X_f,\; y\in X_f\Rightarrow \lambda(x)\triangleleft^{\delta(x)}a\circ\lambda(y)$.\label{en:f}

\end{enumerate}

\end{definition}
The $\lambda$ part of this definition is loosely based on the dual representation $\theta'$ of ordered semigroups (see the second part of Figure~\ref{fig:1}) and is visualised in Figure~\ref{fig:theta}.

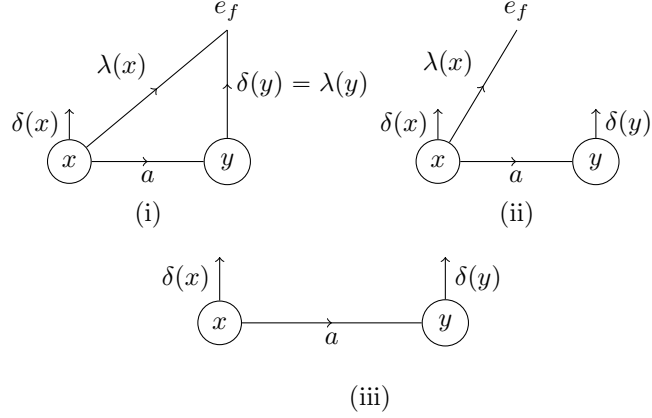
\begin{figure}
    \centering
    \begin{tikzpicture}[scale=.7]

     \node at (1.5,-1){(i)};
     \node at (8.5,-1){(ii)};
        \coordinate (o1) at (3,2.5);
        \node [above] at (o1) {${e_f}$};
        \node[draw, circle] (x1) at (0,0) {$x$};
        \node[draw, circle] (y1) at (3,0) {$y$};
        \coordinate (x1') at (0,1);
        \coordinate (y1') at (3,1);
        
        \path (x1) edge[->-] node[below]{$a$} (y1);
        \path (x1) edge[->] node[left]{$\delta(x)$} (x1');
        \path (x1) edge[->-] node[above left]{$\lambda(x)$} (o1);
        \path (y1) edge[->-] node[right]{$\delta(y) = \lambda(y)$} (o1);
        
        \coordinate (o2) at (8.5,2.5);
        \node[above]  at (o2) {${e_f}$};
        \node[draw, circle] (x2) at (7,0) {$x$};
        \node[draw, circle] (y2) at (10,0) {$y$};
        \coordinate (x2') at (7,1);
        \coordinate (y2') at (10,1);
        
        \path (x2) edge[->-] node[below]{$a$} (y2);
        \path (x2) edge[->] node[left]{$\delta(x)$} (x2');
        \path (x2) edge[->-] node[above left]{$\lambda(x)$} (o2);
        \path (y2) edge[->] node[right]{$\delta(y)$} (y2');

    \end{tikzpicture}
    \begin{tikzpicture}
     \node at (16,-1){(iii)};
     \node[draw,circle] (x3) at (14,0) {$x$};
    \node[draw,circle] (y3) at (17,0){$y$}; 
    \node (x3') at (14,1){};
    \node(y3') at (17,1){};
    \path(x3) edge[->-] node[below] {$a$} (y3);
    \path (x3) edge[->] node[left]{$\delta(x)$}(x3');
     \path(y3) edge[->] node[right] {$\delta(y)$} (y3');
    \end{tikzpicture}
    \caption{$(x,y)\in \theta(a)$ with  (i) $x\in X_i\cup X_f$ and $y \in X_f$,  (ii) $x\in X_i\cup X_f$ and $y \in X_\beta$, and (iii) $x\in X_\beta,\; y\in X_\beta$.  In each case $\delta(x)\btl a\circ\delta(y)$, in case (i) only $\lambda(x)\triangleleft^{\delta(x)} a\circ \lambda(y)$.  In (i) and (ii) if $x\in X_f$  then the $\delta(x)$ and $\lambda(x)$ arrows coincide.}
    \label{fig:theta}
\end{figure}

\begin{lemma}
\label{lem:ref}
Let $\s=(S, \leq, \circ)$ be a structure where $\circ$ is associative, $\leq$ is a partial order, $\s\models\sigma$ and suppose there is an identity $e\in\s$.   Let $\theta$ be from Definition~\ref{def:theta}.  Then 
 $a  \leq b \in \mathcal{S}$ if and only if   $a^\theta\sqsubseteq b^\theta$.
\end{lemma}
\begin{proof}
Assume $a\not\leq b$, so either $\neg b\blacktriangleleft a$ or $\neg a\triangleleft^b b$, by $\sigma$.  In the former case $(b_\beta, e_\beta)\in b^\theta$ but $b_\beta\not\in \dom{a^\theta}$.  Otherwise $b\blacktriangleleft a$ and the latter case holds, but then $((b_i, a_i), e_f)\in a^\theta\setminus b^\theta$.  Either way, $a^\theta\not\sqsubseteq b^\theta$.

Now suppose $a\leq b$.  First we check that $\dom{b^\theta}\subseteq \dom{a^\theta}$.  
If $x\in \dom{b^\theta}$ there is $y\in X$ where $(x, y)\in b^\theta$.  It follows that $\delta(x)\blacktriangleleft b\blacktriangleleft a$, so $(x, e_\beta)\in a^\theta$ and $x\in \dom{a^\theta}$.  Secondly, if $x\in \dom{b^\theta}$ (so $\delta(x)\blacktriangleleft b$)
and $(x, y)\in a^\theta$ we know that \eqref{en:delta}--\eqref{en:f} hold for $a$, in particular $\delta(x)\blacktriangleleft a\circ \delta(y)$.  It follows that $\delta(x)\blacktriangleleft b\circ \delta(y)$,  by Lemma~\ref{lem:trans}\eqref{en:bx}, as required by \eqref{en:delta}. 
Conditions \eqref{en:i},\eqref{en:b} remain true for $b^\theta$.  For \eqref{en:f} if $x\in X_f$ then $\lambda(x)\triangleleft^{\delta(x)}a\circ \lambda(y)\triangleleft^{\delta(x)}b\circ\lambda(y)$, by Lemma~\ref{lem:trans}\eqref{en:tx}.
Hence $(x, y)\in b^\theta$, thus $a^\theta\sqsubseteq b^\theta$.

\end{proof}

\begin{lemma}
\label{lem:comp} Let $\s=(S, \leq, \circ)$ be a structure where $\circ$ is associative, $\leq$ is a partial order, $\s\models\sigma$ and let $\theta$ be from from  Definition~\ref{def:theta}.  For any $a,b \in \mathcal{S}$, we have  $(a \circ b)^\theta = a^\theta;b^\theta$.
\end{lemma}

\begin{proof}
First, let's show that $ a^\theta;b^\theta \subseteq (a;b)^\theta$. Take any $(x, y) \in a^\theta$ and $(y,z)\in b^\theta$.  We have $\delta(x)\blacktriangleleft a\circ\delta(y)$ and $\delta(y)\btl b\circ \delta(z)$, so $\delta(x)\btl a\circ b\circ\delta(z)$, by Lemma~\ref{lem:trans}\eqref{en:bx}.  By \eqref{en:i},  $z\not\in X_i$ and by \eqref{en:b} if $x\in X_\beta$ then $y\in X_\beta$ and then $z\in X_\beta$. Also by \eqref{en:i}  $y\not\in X_i$ so if $z\in X_f$ then $y\in X_f$ and $\lambda(y)=\delta(y)$.  Then $x\in X_i\cup X_f$ and $\lambda(x)\triangleleft^{\delta(x)} a\circ\lambda(y),\; \lambda(y)\triangleleft^{\\   \lambda(y)}b\circ\lambda(z)$, so $\lambda(x)\triangleleft^{\delta(x)}a\circ b\circ\lambda(z)$, by Lemma~\ref{lem:trans}\eqref{en:tx}.  Hence $(x, z)\in(a\circ b)^\theta$.

Conversely, to show that $(a \circ b)^\theta \subseteq a^\theta;b^\theta$, take any $(x,y)\in (a \circ b)^\theta$.    By \eqref{en:i} $y\not\in X_i$.  If $y\in X_\beta$ let $z=(b\circ\delta(y))_\beta\in X_\beta$ (Figure~\ref{fig:witness} right) and then $(x, z)\in a^\theta,\;(z, y)\in b^\theta$.  Otherwise, $y\in X_f$ and we let $z=(b\circ\lambda(y))_f\in X_f$ (Figure~\ref{fig:witness} left) and again we have $(x, z)\in a^\theta,\; (z, y)\in b^\theta$, as required.

\begin{figure}
\begin{center}
 \begin{tikzpicture}
    \coordinate(o) at (3,2.3);
    \node[draw, circle](s) at (0,0){$x$};
    \node[draw, circle](t) at (2.25,0){$z_f$};
    \coordinate(s') at (0,1);
    \node[draw,circle](u) at (4.5,0){$y$};
    
    \path (s) edge[->-] node[left]{$\lambda(x)\,\,\,$} (o);
    \path (t) edge[->-] node[below left, pos=0.4]{$b \circ \lambda(y)$} (o);
    \path (s) edge[->-, dashed] node[below]{$a$} (t);
    \path (s) edge[->] node[left]{$\delta(x)$}(s');
    \path (t) edge[->-, dashed] node[below]{$b$} (u);
    \path (u) edge[->-] node[right]{$\,\lambda(y)=\delta(y)$} (o);
    \path (s) edge[->-, bend right] node[below]{$a\circ b$} (u);

    \coordinate(o1) at (7,3);
    \node[draw, circle](s) at (6,0){$x$};
    \node[draw, circle](t) at (7.6,0.7){$z_\beta$};
    \coordinate(s') at (6,1);
    \coordinate(t') at (7.6,1.7);
    \node[draw,circle](u) at (9,0){$y$};
    \coordinate(u') at (9,1);
    
    \node[above] at (o) {$e_f$};

    \path (s) edge[->-, dashed] node[below right]{$a$} (t);
    \path (s) edge[->] node[left] {$\delta(x)$} (s');

    \path (t) edge[->] node[right]{$b\circ \delta(y)$}(t');
    \path (t) edge[->-, dashed] node[below left]{$b$} (u);
    \path (u) edge[->] node[right]{$\delta(y)$} (u');
    \path (s) edge[->-, bend right] node[below]{$a\circ b$} (u);
\end{tikzpicture}  
\end{center}
\caption{\label{fig:witness} Witness  for $(x,y)\in (a \circ b)^\theta$ where $x\in X_i\cup X_f$ and  $y\in X_f$ (left), $x\in X,\;y\in X_\beta$ (right).}
\end{figure}
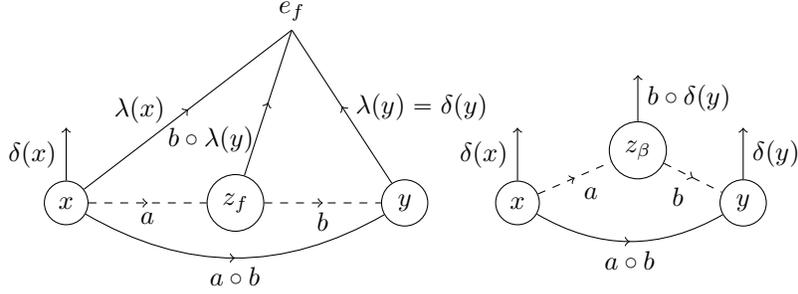
\end{proof}

\begin{theorem} \label{thm:ax}
$R(\dleq, ;)$ is axiomatised by partial order, associativity and $\{\sigma_n:n<\omega\}$.
Finite structures $\mathcal{S} \in R(\dleq, ;)$ are representable over a finite base $X$ with $|X| \leq (1+|\s|)^2+2\cdot(1+|\s|)$.
\end{theorem}

\begin{proof}
Soundness of partial order, associativity is clear, soundness of   $\sigma_n$ is from Lemma~\ref{lem:sound}.    For completeness, take any associative, partially ordered $(\leq, \circ)$-structure $\s\models\set{\sigma_n:n<\omega}$.  We may define $\s'$ be adding a new identity $e$ to $\s$ unordered with other elements.   By  Lemmas~\ref{lem:ref} and \ref{lem:comp}, the map $\theta$ of Definition~\ref{def:theta} is a $(\sqsubseteq, ;)$-representation of $\s'$, hence it restricts to a $(\sqsubseteq, ;)$-representation of $\s$.  The representation $\theta$ has 
base contained in a disjoint union of a copy of $(\s')^2$ and two copies of $\s'$.
\end{proof}

\section{$R(\sqsubseteq, ;)$ is not Finitely Axiomatisable}

\begin{definition}
Let $n < \omega,\; N = 1+2^{n}$ and let $\mathcal{S}_n$ be a $(\leq, \circ)$-structure whose underlying set $S_n$ has $3+3N$ elements

$$S_n = \{0,b,c\} \cup \{a_i, a_ib, a_ic :  i < N\}$$
where  composition $\circ$ is defined by   $a_i\circ b = a_ib, \;a_i \circ c = a_ic$ (all $i < N$) and all other compositions result in $0$, \/ and 
the refinement operation is defined as the reflexive closure of

$$\{(s,0) : s \in S_n\} \cup \{(a_{i  + 1}b ,a_i), (a_i, a_{i  + 1}c), (a_{i}b, a_{i}c):i<N \}$$ 
where here and below the operator $+$ denotes addition modulo $N$.
\end{definition}

 Observe that $\circ$ is associative and $\leq$ is a partial order.
 \begin{lemma}
 For $n\geq 2$, \/ $\s_n$ is not representable, but $\s_n\models\sigma_k$ for $k< n$.
 \end{lemma}
 
 \begin{proof}
 Since $s\leq 0$   we have $0\btl_0 s$. Also, for $i<N$,  since $a_{i+1}\circ b\leq a_i$ we have 
 $a_i\btl_0 a_{i+1}b\btl_0 a_{i+1},\; a_i\btl_0 a_{i+1}$, so $a_i\btl_k a_{i+2^{k}}$ for $k\geq 1$, using Lemma~\ref{lem:trans}\eqref{en:trans}.  Hence
 $\set{a_i, a_ib:i<N}$ is a clique of $\btl_n$, but for $k<n$ we do not have $a_{i+1}\btl_k a_i$ nor do we have $a_{i+1}c\btl_k a_ib$. 
 
 For $\triangleleft$, we have
 \begin{itemize}
     \item 
$t \triangleleft_0^s u$ iff $t \leq u$ and $s = u$, i.e.
  $s \triangleleft^s_0 s, \; s\triangleleft^0_00$ (all $s$),  $a_{i+1}b\triangleleft_0^{a_i} a_{i},\; 
  a_i\triangleleft_0^{a_{i+1}c} a_{i+1}c$ and  $a_ib \triangleleft_0^{a_{i}c} a_{i}c$  (all $i<N$), but $\triangleleft_0$ holds in no other cases.
  
  \item   
  Since $a_{i+1}b\circ c=0$ and $a_{i+1}b\triangleleft_0^{a_{i}}a_{i}$, it follows by Lemma~\ref{lem:trans}\eqref{en:tx} that $0\triangleleft_1^{a_i} a_{i}c$, similarly, $0\triangleleft_1^{a_i} a_{i}b$.
  Also by Lemma~\ref{lem:trans}\eqref{en:tx}, since $a_{i}\triangleleft_0^{a_{i+1}c}a_{i+1}c$ and $a_{i+1}c\circ b=0$ we get $a_{i}b\triangleleft_1^{a_{i+1}c}0$, similarly $a_ic\triangleleft_1^{a_{i+1}c}0$. And  from $a_i \triangleleft^{a_i}_{0} a_i$ and $c \triangleleft^c_0 c$ we get $a_ic \triangleleft^{a_i}_1 a_ic$, similarly $a_ib\triangleleft_1^{a_i}a_ib$.   The only non-zero products are $a_i\circ b$ and $a_i\circ c$, so the only remaining case of $\triangleleft_1$ we obtain from Lemma~\ref{lem:trans}\eqref{en:tx} is  $0 \triangleleft^s_1 0$, which follows
  since $s \triangleleft^s_0 s$, for all $s \in S_n$ and $0 \triangleleft_0^0 0$.  
  By Lemma~\ref{lem:trans}\eqref{en:ss'},    from $a_{i+1}b\triangleleft_0^{a_i}a_i$ we get $a_{i+1}b\triangleleft_1^{s}a_i$ for $s\btl_0a_i$.
      This concludes the exhaustive enumeration of elements in $\triangleleft_1$, not covered by $\triangleleft_0$.

\item If $a\triangleleft_1^sb$ and $s'\btl_1 s$ we get $a\triangleleft_2^{s'}b$, in particular $0\triangleleft_2^{a_ic}a_ic$.

 \item  Since $a_ib\triangleleft_1^{a_{i+1}c}0\triangleleft_2^{a_{i+1}c}a_{i+1}c$, it follows by Lemma~\ref{lem:trans}\eqref{en:trans} that  $a_{i}b\triangleleft_3^{a_{i+1}c}a_{i+1}c$.
 
 \item The remaining cases of $\triangleleft$ can be enumerated as follows..  We have  $0\triangleleft^{s} a_{i+1}c, \; 0\triangleleft^{s} a_{i+1}b, \; a_ic \triangleleft^s a_ic$ for $s\btl a_0$, by Lemma~\ref{lem:trans}\eqref{en:bx}. Additionally, since $a_{i+1}b \triangleleft^s a_i$, we get  $0 \triangleleft^s a_i$, by Lemma~\ref{lem:trans}\eqref{en:trans}.   Also by Lemma~\ref{lem:trans}\eqref{en:trans}, for any $s \in S_n$ since  $s \triangleleft^0 0$ and $0 \triangleleft^0 a_ib, \; 0 \triangleleft^0 a_ic, \; 0 \triangleleft^0 a_i$, we have $s \triangleleft^0 a_ib, \; s \triangleleft^0 a_ic, \; s \triangleleft^0 a_i$, and if  $a \in \{a_ib, a_ic:i<N\}$, \/ and $b \in\set{a_i,a_ib, a_ic:i<N}$ we have $a \triangleleft^{a_{i+1}c} b$.
 \end{itemize}
 
This covers all triples $(a, s, b)$ where $a\triangleleft^{s}b$.  
 It follows that $\s_n\not\models\sigma_{n+1}$ for $n\geq 2$, since $a_{i+1}c\btl_1 a_{i+1}\btl_n  a_ib,\;a_ib\triangleleft_3^{a_{i+1}c}a_{i+1}c$ but $\s_n\not\models a_{i}b\leq a_{i+1}c$.    By Theorem~\ref{thm:ax}, $\s_n$ is not representable. The only cases where $a\triangleleft^bb$ and $a\not\leq b$ are $a_{i}b\triangleleft^{a_{i+1}c}a_{i+1}c$, but for $k<n$ we do not have $a_{i+1}c\btl_k a_ib$, hence $\s_n\models\sigma_k$.

 \end{proof}

\begin{theorem}\label{thm:nfa}
$R(\sqsubseteq, ;)$ cannot be defined by finitely many axioms.
\end{theorem}
\begin{proof}
Each structure $\s_n\not\in R(\sqsubseteq, ;)$.  For any $k<\omega$ almost all $\s_n$ satisfy $\sigma_k$ (in fact, all $\s_n$ where $n > k$) and they are all associative and partially ordered, hence  any non-principal ultraproduct $\s=\Pi_U\s_n$ is associative, partially ordered and satisfies all $\sigma_k$s, so  by Theorem~\ref{thm:ax}, $\s\in R(\sqsubseteq, ;)$.  By \Los' theorem, $R(\sqsubseteq, ;)$ has no finite axiomatisation.
\end{proof}

\section{Finite axiomatisability and representability}
For any relation algebra signature $\Sigma$, the representation class $R(\Sigma)$ may be finitely axiomatisable or not, and it may be that finite representable structures have finite representations or not.  All four combinations of these two properties are possible. 

\begin{theorem}
The representation class $R(\Sigma)$ is finitely axiomatisable, and finite structures in $R(\Sigma)$ have finite representations, according to the following incomplete table.
\[
\begin{array}{l | ll}
& \mbox{fin. ax.}&\mbox{not fin. ax.}\\
\hline
\mbox{fin. rep}&(\subseteq,D,R,{}^\smile,;)&(\sqsubseteq, ;) \\
&(\subseteq,;),(\sqsubseteq, *)&\\ \\
\mbox{not fin. rep.}&(\cap,;)&
(\cap,\cup,;)\subseteq\Sigma\\
&&(\subseteq, \setminus,;)\subseteq\Sigma
\end{array} 
\]
where $\Sigma' \subseteq \Sigma$ signifies the language characterised by $\Sigma$ being an expansion of the language characterised by $\Sigma'$.
\end{theorem}

\begin{proof}
Finite axiomatisability of $R(\subseteq, D, R, {}^\smile,;)$ is proved in \cite{Bre77} and the finite representation  property for this signature is proved in \cite{HM13}.
 Both $R(\subseteq, ;)$ and $R(\sqsubseteq, *)$ are defined by the axioms of ordered semigroups and have the finite representation property \cite{zareckii1959representation,HMS:20}.

The finite representation property is proved for $R(\sqsubseteq, ;)$  in  Theorem~\ref{thm:ax}, non-finite axiomatisability is proved in Theorem~\ref{thm:nfa}. The failure of the finite representation property for signatures containing $(\cap,;)$ is proved in \cite{neuzerling2017representability}, finite axiomatisability of  $R(\cap, ;)$ is proved in Proposition~\ref{prop:fin} below.

For the final quadrant of the diagram, if the representation problem for finite structures in $R(\Sigma)$ is undecidable, we know that there can be no finite axiomatisation, and since the set of formulas valid over $R(\Sigma)$ is recursively enumerable the finite representation property cannot hold.  The representation problem for finite structures is proved undecidable for signatures containing $(\cap,\cup, ;)$ in \cite{HJ12} and for signatures containing $(\subseteq, -, ;)$, where negation is interpretted as complementation relative to a universal relation
 $X\times X$, in \cite{Neuz}. We extend that result to prove failure of the finite representation property  for representations  where $-$ denotes complementation relative to an arbitrary maximal binary relation in Proposition~\ref{prop:neg}, below. 
 
\end{proof}

\begin{proposition}\label{prop:fin}
$R(\cap, ;)$ is finitely axiomatisable.
\end{proposition}

\begin{proof}
A $(\cap, ;)$-representable $(\cdot, \circ)$-structure clearly satisfies the semilattice laws, associativity and monotonicity.  Conversely, in a representation game played over an associative, monotonic semilattice $\c S$,\/ $\exists$ plays a sequence of networks  --- graphs $N$ whose edges are labelled by upward closed subsets of $\c  S$, such that $N(x, y);N(y, z)\subseteq N(x, z)$ for all $x, y, z\in N$.  [See Definition 7.7 of  \cite{HH:book} for more details of a representation  game for the full  signature of relation algebra, and Chapter 9 for representation games in  a more general setting.] In the initial round suppose $\forall$ picks $a\neq b$.  By antisymmetry either $a\not\leq b$ or $b\not\leq a$, without loss assume the former. 
$\exists$ plays a network $N_0$ with two nodes  labelled $N(x, y)=a^\uparrow$, all other edges have empty labels, note that $b\not\in N(x, y)$. In a subsequent round let $N$ be the current network.  $\exists$ adds a single new node $z$ and lets $N'(w,z)=N(w,x);a^\uparrow,\; N'(z, w)=b^\uparrow;N(y, w)$ for all $w\in N$ to define $N'$.  Edges  within $N$ are not refined.    If $u, w\in N$ then $N'(u, z);N'(z, w)=N(u, x);a^\uparrow;b^\uparrow;N(y, w)\subseteq N(u, x);N(x, y);N(y,w)$, since $a;b\in N(x, y)$, using associativity, left and right monotonicity (see Figure~\ref{fig:meetCompInduction}). It is easily seen that $N'$ is a consistent network, a legal response to $\forall$'s move not refining the initial edge.   It follows that $\c S\in R(\cap, ;)$

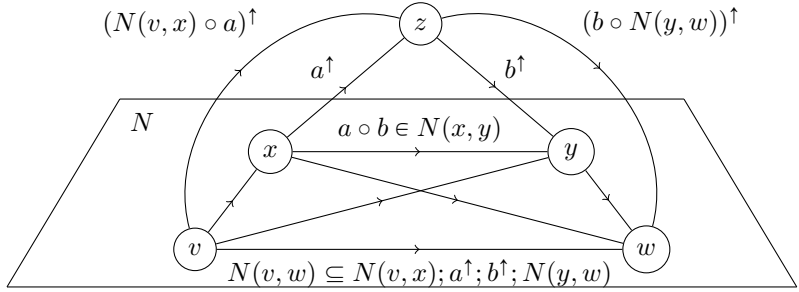
\begin{figure}
    \centering
    \begin{tikzpicture}
        \node[circle, draw] (x) at (0.5,1.8){$x$};
        \node[circle, draw] (y) at (4.5,1.8){$y$};
        \node[circle, draw] (z) at (2.5, 3.5){$z$};
        \node[circle, draw] (v) at (-0.5,0.5){$v$};
        \node[circle, draw] (w) at (5.5,0.5){$w$};
        
        \draw (-3,0) --  (7.5, 0);
        \draw (-3,0) --  (-1.5, 2.5);
        \draw (6,2.5) --  (-1.5, 2.5);
        \draw (6,2.5) --  (7.5, 0);
        \node at (-1.2,2.2) {$N$};
        
        \path (x) edge[->-] node[above]{$a\circ b \in N(x,y)$} (y);
        \path (x) edge[->-] node[above left]{$a^\uparrow$} (z);
        \path (z) edge[->-] node[above right]{$b^\uparrow$} (y);
        \path (v) edge[->-, bend left=60] node[above left, pos=0.6]{$(N(v,x)\circ a)^\uparrow$} (z);
        \path (z) edge[->-, bend left=60] node[above right, pos=0.4]{$(b \circ N(y,w))^\uparrow$} (w);
        \path(v) edge[->-]  (x);
        \path(v) edge[->-]  (y);
        \path(x) edge[->-]  (w);
        \path(y) edge[->-]  (w);
        \path(v) edge[->-] node[below]{$N(v,w) \subseteq N(v,x);a^\uparrow;b^\uparrow;N(y,w)$} (w);
    \end{tikzpicture}
    \caption{Node Addition in a Representation Game for $(\cdot, \circ)$}
    \label{fig:meetCompInduction}
\end{figure}
\end{proof}

 Let $\s$ be a $(\leq, -, \circ)$-structure.  A $(\subseteq, \setminus, ;)$-representation of $\s$  over base $X$ is a map $\theta:\s\rightarrow \wp(X\times X)$ such that for all $a, b\in\s$,
 \begin{itemize}
     \item $a\leq b \rightarrow a^\theta\subseteq b^\theta$,
     \item $(x, y)\in a^\theta\rightarrow (x, y)\in\Delta(b^\theta, (-b)^\theta)$ (the symmetric difference of $b^\theta$ and $(-b)^\theta$),
     \item $(x, y)\in (a\circ b)^\theta\leftrightarrow \exists z((x, z)\in a^\theta\wedge (z, y)\in b^\theta)$.
 \end{itemize}
 According to this definition, $-$ is represented as complementation in the union of all the represented binary relations.

\begin{proposition}\label{prop:neg}
For any $RA$ reduct $(\subseteq,\setminus,;) \subseteq \Sigma$, the  representation class  $R(\Sigma)$ fails to have the finite representation property for finite representable structures.
\end{proposition}

\begin{proof}
The point algebra $\c P$
 is a relation algebra whose boolean part has  three atoms $e, l, g$ (so 8 elements, $0, e, l, g, -e, -l, -g, 1$), where $e$ is the identity, the converse of $l$ is $g$, composition for atoms is given by
\[
\begin{array}{l|lll}
\circ &e&l&g\\
\hline
e&e&l&g\\
l&l&l&1\\
g&g&1&g
\end{array}
\]
and the operators extend to arbitrary elements by additivity.  A representation of $\c P$ over $\rats$ may be obtained by representing $e, l, g$ as the identity $\set{(q, q):q\in\rats}$, less than $\set{(q, q'): q<q'}$ and greater than, respectively.  It follows that the reduct of $\c P$ to $(\leq, -, \circ)$ is $(\subseteq, \setminus, ;)$-representable.  We claim it has no finite $(\subseteq, \setminus, ;)$-representation.  

Let $\theta$ be any $(\subseteq, \setminus, ;)$-representation of $\c P$ over the base $X$.\\
{\bf Claim 1:}  If $(x, y)\in g^\theta$ then $x\neq y$. 
To prove the claim, suppose for contradiction that there is a point $x \in X$ with $(x,x) \in g^\theta$. As $g \leq -e$, $(x,x) \in g^\theta\subseteq (-e)^\theta$. And since $g=g\circ e$, there exists a $y$ s.t. $(x,y) \in g^\theta, \;(y,x) \in e^\theta$. Since $e\circ g = g$, we also have $(y,x) \in g^\theta$. But $e \leq -g$, so $(y,x) \in (-g)^\theta$. Since $(y, x)\in g^{\theta}$ we  have reached a contradiction and proved claim 1.\\
{\bf Claim 2:} For $n\geq 0$ there is $x\in X$ and distinct points $y_0, \ldots, y_n\in X$ such that for all $i\leq n$ we have $(x, y_i)\in (-g)^\theta$ and for all $i<j\leq n$ we have $(y_j, y_i)\in g^\theta$.  See Figure~\ref{fig:(leq,-,;) induction}.
Claim 2 is proved by induction over $n$.  
For the base case, $n=0$, since $(-g)\not\leq 0$ there are $x, y_0$ where $(x, y_0)\in (-g)^\theta$.  Assume the hypothesis for some $n\geq 0$.    Since $(x, y_n)\in (-g)^\theta$ and $(-g)\leq 1= (-g)\circ g$, there must be $y_{n+1}\in X$ where $(x, y_{n+1})\in (-g)^\theta$ and $(y_{n+1}, y_n)\in g^\theta$.    Since $(y_n, y_i)\in g^\theta$ it follows that $(y_{n+1}, y_i)\in (g\circ g)^\theta = g^\theta$, for $i\leq n$.  By the previous claim, $y_{n+1}$ is distinct from $y_i$, for $i\leq n$, as required.  This proves  claim 2.

Since $X$ contains a set of $n$ distinct points, for all $n<\omega$, it follows that $X$ must be infinite.

\begin{figure}
    \centering
    \begin{tikzpicture}
   
    \node[circle, draw] (x) at(5,0) {$x$};
    \node[circle, draw] (y0) at(9,0) {$y_i$};
    \node[circle, draw] (yi) at(7,1) {$y_n$};
    \node[circle, draw] (yn) at(7,3) {$y_{n+1}$};

    \draw[->-] (x) -- node[above]{$-g$} (y0);
    \draw[->-] (x) -- node[above]{$-g$} (yi);
    \draw[->-] (x) -- node[above left]{$-g$} (yn);
    \draw[->-] (yn) -- node[right]{$g$} (yi);
    \draw[->-] (yn) -- node[above right]{$g$} (y0);
    \draw[->-] (yi) -- node[above]{$g$} (y0);
    \draw[dashed] (8.5,-0.5) rectangle (10.5,1);
    \node at (9.7,0.75){$0\leq i<n$};
    
  \end{tikzpicture} 
  
    \caption{Induction showing a new node is needed for representation of $P$}
    \label{fig:(leq,-,;) induction}
\end{figure}
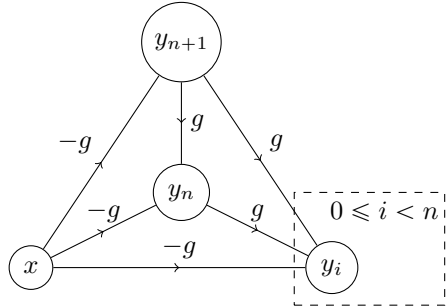
\end{proof}

\section{Demonic Lattice and Semilattice}

We have seen in the introduction that demonic join $\sqcup$ is the meet operation for demonic refinement $\sqsubseteq$.   A demonic meet $\sqcap$, acting as a least upper bound of its two arguments, may not in general be defined, as there are binary relations having no common lower bound at all.  If a point $x$ is in the domain of two binary relations $R, S$, but  not in the domain of $R\cap S$, then any lower bound of $R, S$ would be below the intersection $R\cap S$, hence $x$ would be outside its domain, yet in order to be a lower bound its domain should contain both the domain of $R$ and the domain of $S$, a contradiction.  This problem could solved be adding a single new point $\infty$ to the base $X$ of the representation $\theta$ and letting $\theta'(R)=\theta(R)\cup\set{(x, \infty): x\in \dom{R}}$  to obtain an alternative representation of the refinement algebra, with $\sqsubseteq$-least element $\set{(x, \infty): x\in X\cup\set\infty}$.  Over such a representation, a greatest lower bound may be defined by
\begin{align*}
    R\sqcap S = & \set{(x, y):(x, y)\in R,\; x\not\in \dom{S}}\\
    & \cup (R\cap S) \\
    & \cup \set{(x, y):(x, y)\in S,\;x\not\in \dom{R}}
\end{align*}
Hence,  every representable $(\sqsubseteq,;)$-structure embeds into a representable $(\sqcap,$ $\sqcup, ;)$-structure forming a distributive lattice with composition.  We expect that additional properties are required  to ensure that such a representation exists.  

\begin{problem}
Is the class of representable semigroups with demonic semilattice $R(\sqcup, ;)$ finitely axiomatisable and are the finite structures in $R(\sqcup, ;)$ representable over  finite bases?
\end{problem}

\begin{problem}
Find axioms for the class of all $(\sqcup, \sqcap, ;)$-structures of binary relations with demonic join and meet under composition.
\end{problem}

\bibliography{ref}
\bibliographystyle{alpha}

\end{document}